\documentclass{article}
\usepackage[utf8]{inputenc}
\usepackage{hyperref}
\usepackage{graphicx}
\usepackage{amsmath}
\usepackage{amssymb}
\usepackage{amsfonts}
\usepackage{amsthm}
\usepackage{MnSymbol}
\usepackage{thm-restate}
\usepackage{cleveref}
\usepackage{caption}
\usepackage{subcaption}
\usepackage[ruled,vlined,linesnumbered]{algorithm2e}
\newcommand{\chaincover}{\mathcal{C}}
\newcommand{\pathcover}{\mathcal{P}}
\newcommand{\partition}{\mathcal{K}}
\newcommand{\sizesplit}{SIZE-SPLIT}
\newcommand{\search}{SEARCH}
\newcommand{\merge}{MERGE}
\newcommand{\splitt}{SPLIT}
\newcommand{\join}{JOIN}

\newcommand{\some}{SOME}
\newcommand{\select}{SELECT}
\newcommand{\flowG}{\mathcal{G}}
\newcommand{\flowV}{\mathcal{V}}
\newcommand{\flowE}{\mathcal{E}}

\newtheorem{lemma}{Lemma}
\newtheorem{corollary}{Corollary}
\crefname{algocf}{alg.}{algs.}
\Crefname{algocf}{Algorithm}{Algorithms}

\title{Minimum Chain Cover in Almost Linear Time\thanks{This work was funded by the Academy of Finland
(grants No. 352821, 328877). I am very grateful to Alexandru~I.~Tomescu and Brendan Mumey for initial discussions on flow decomposition of a min flow representing an MPC.}}
\author{Manuel C\'aceres\thanks{Department of Computer Science, University of Helsinki, Finland, \texttt{manuel.caceresreyes@helsinki.fi}.}}
\date{}

\begin{document}

\maketitle
\begin{abstract}
    A minimum chain cover (MCC) of a $k$-width directed acyclic graph (DAG) $G = (V, E)$ is a set of $k$ chains (paths in the transitive closure) of $G$ such that every vertex appears in at least one chain in the cover. 
    
    The state-of-the-art solutions for MCC run in time $\tilde{O}(k(|V|+|E|))$ [M\"akinen et~at., TALG], $O(T_{MF}(|E|) + k|V|)$, $O(k^2|V| + |E|)$ [C\'aceres et~al., SODA 2022], $\tilde{O}(|V|^{3/2} + |E|)$ [Kogan and Parter, ICALP 2022] and $\tilde{O}(T_{MCF}(|E|) + \sqrt{k}|V|)$ [Kogan and Parter, SODA 2023], where $T_{MF}(|E|)$ and $T_{MCF}(|E|)$ are the running times for solving maximum flow (MF) and minimum-cost flow (MCF), respectively.
    
    In this work we present an algorithm running in time $O(T_{MF}(|E|) + (|V|+|E|)\log{k})$. By considering the recent result for solving MF [Li et~al., FOCS 2022] our algorithm is the first running in almost linear time. Moreover, our techniques are deterministic and derive a deterministic near-linear time algorithm for MCC if the same is provided for MF.
    
    At the core of our solution we use a modified version of the mergeable dictionaries [Farach and Thorup, Algorithmica], [Iacono and \"Ozkan, ICALP 2010] data structure boosted with the \sizesplit{} operation and answering queries in amortized logarithmic time, which can be of independent interest.
    
\end{abstract}

\section{Introduction}
Computing a minimum-sized set of chains covering all vertices of a DAG $G = (V, E)$ is a well known poly-time solvable problem~\cite{dilworth2009decomposition,fulkerson1956note}, with many applications in widespread research fields such as bioinformatics~\cite{trapnell2010transcript,caceres2021safety,cotumaccio2021indexing,caceres2022width,CJ23,rizzo2023chaining}. Here we call such an object a \emph{minimum chain cover} (an MCC) $\chaincover$ containing $k$ chains $\chaincover = \{C_1,\ldots,C_k\}$, which are paths in the transitive closure of $G$. The \emph{size} $k$ of an MCC is known as the \emph{width} of $G$ and equals the maximum number of pairwise unreachable vertices (antichain) of $G$, by Dilworth's theorem~\cite{dilworth2009decomposition} on partially ordered sets (posets).\\

\par{\noindent\textbf{The history of MCC.}} It was Fulkerson~\cite{fulkerson1956note} in the 1950s the first to show a poly-time algorithm for posets (transitive DAGs). His algorithm reduces the problem to finding a maximum matching in a bipartite graph with $2|V|$ vertices and $|E|$ edges, and thus can be solved in $O(|E|\sqrt{|V|})$ time by using the Hopcroft-Karp algorithm~\cite{hopcroft1973n}\footnote{Recent fast solutions for maximum matching do not speed up this approach as one needs to compute the transitive closure of the DAG first.}. Improvements on these ideas were derived in the $O(|V|^2 + k\sqrt{k}|V|)$ and $O(\sqrt{|V|}|E| + k\sqrt{k}|V|)$ time algorithms of Chen and Chen~\cite{chen2008efficient,chen2014graph}, and the $O(k|V|^2)$ time algorithm of Felsner et~al.~\cite{felsner2003recognition}. In the same article, Felsner et~al. showed a combinatorial approach to compute a \emph{maximum antichain} (MA) in near-linear time for the cases $k=2,3,4$, which was latter generalized to $O(f(k)(|V|+|E|))$~\cite{caceres2021a}, $f(k)$ being an exponential function. State-of-the-art approaches improve exponentially on its running time dependency on $k$. These approaches solve the strongly related problem of \emph{minimum path cover} (MPC). An MPC $\pathcover$ is a minimum-sized set of \emph{paths} covering the vertices of $G$, and thus it is also a valid chain cover. Moreover, since every MCC can be transformed into an MPC (by connecting consecutive vertices in the chains), the \emph{size} of an MPC also equals the width $k$.\\

\par{\noindent\textbf{The state-of-the-art for MCC.}} M\"akinen et~al.~\cite{makinen2019sparse} provided an algorithm for MPC, running in time $O(k(|V|+|E|)\log{|V|}) = \widetilde{O}(k(|V|+|E|))$ while C\'aceres et~al.~\cite{caceres2022sparsifying} presented the first $O(k^2|V|+|E|)$ parameterized linear time algorithm. Both algorithms are based on a classical reduction to \emph{minimum flow}~\cite{ntafos1979path}, which we will revisit later. In the same work, C\'aceres et~al.~\cite{caceres2022sparsifying} showed how to compute an MPC in time $O(T_{MF}(|E|) + ||\pathcover||)$ and a MA in time $O(T_{MF}(|E|))$, where $T_{MF}(|E|)$ is the running time for solving \emph{maximum flow} (MF) and $||\pathcover||$ is the total length of the reported MPC. As such, by using the recent result for MF of Chen et~al.~\cite{chen2022maximum} we can solve MPC and MA in (almost) optimal (input+output size) time. However, the same does not apply to MCC as the total length of an MCC can be exactly $|V|$ (e.g. by removing repeated vertices) while the total length of an MPC can be $\Omega(k|V|)$ in the worst case, as shown in \Cref{fig:worst-case-mpc}. 

The $k|V|$ barrier was recently overcame by Kogan and Parter~\cite{kogan2022beating,kogan2023faster} by reducing the total length of an MPC. They obtain this improvement by using \emph{reachability shortcuts} and by devising a more involved reduction to \emph{minimum cost flow} (MCF). Their algorithms run in time $\widetilde{O}(|E|+|V|^{3/2})$~\cite{kogan2022beating} and $\widetilde{O}(\sqrt{k}|V| + |E|^{1+o(1)})$~\cite{kogan2023faster} using the MCF algorithms of Bran et~al.~\cite{van2021minimum} and Chen et~al.~\cite{chen2022maximum}, respectively.\\

\begin{figure}
      \centering
      \includegraphics[width=130mm]{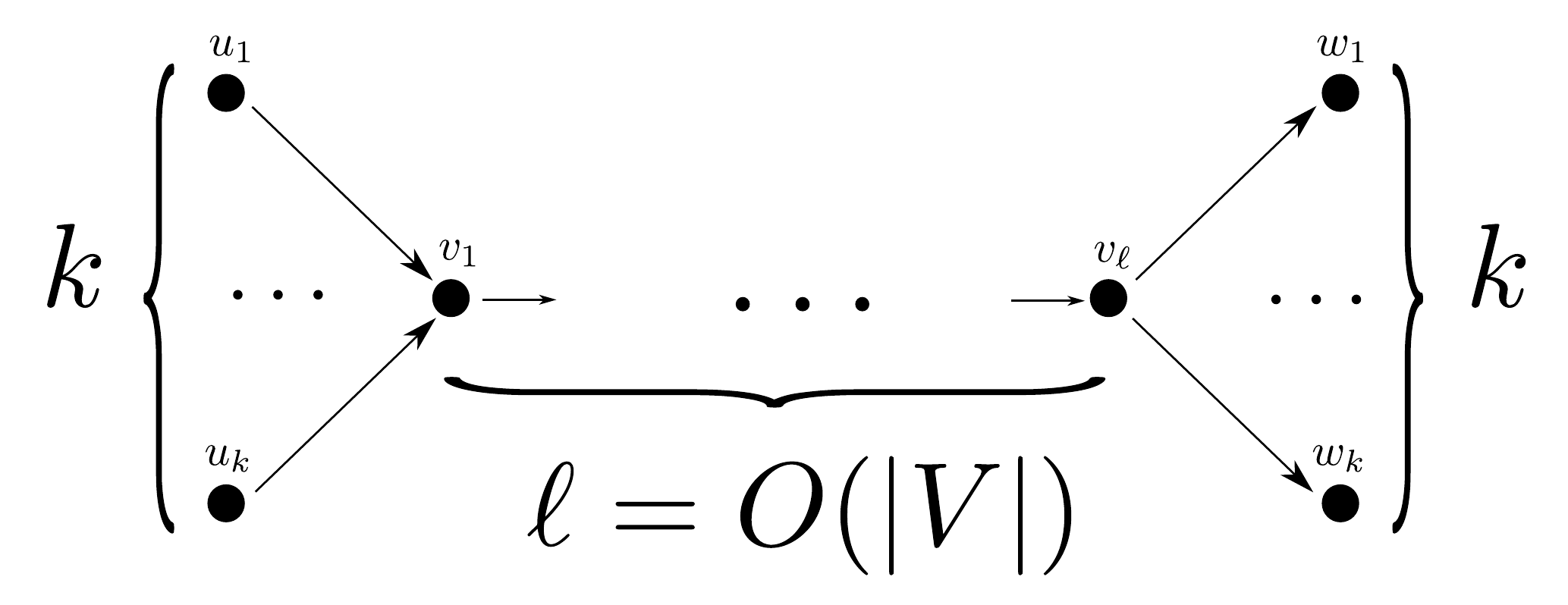}
      \caption{Example $k$-width DAG where every MPC $\pathcover$ has total length $||\pathcover|| = \Omega(k|V|)$. Indeed, every path in an MPC must start from some $u_i$ and traverse the middle path $v_1,\ldots,v_\ell$ until reaching some $w_j$, since otherwise it is not possible to cover the rest of the graph minimally. Moreover, if $k = \ell = |V|/3$, then $||\pathcover|| = \Omega(|V|^2)$. On the other hand, there is always an MCC $\chaincover$ of total size $||\chaincover|| = |V|$ (in this case only one of the chains cover the vertices of the middle path).}
      \label{fig:worst-case-mpc}
  \end{figure}
  
  In this paper we present an algorithm for MCC improving its running time dependency on $k$ exponentially w.r.t. the state-of-the-art.
  
  \begin{restatable}{theorem}{mainTheorem}\label{thm:mainTheorem}
   Given a $k$-width DAG $G=(V,E)$ we can compute a minimum chain cover in time $O(T_{MF}(|E|)+ (|V|+|E|)\log{k})$, where $T_{MF}(|E|)$ is the time for solving maximum flow.
  \end{restatable}
  
  Thus by applying the flow algorithm of Chen et~al.~\cite{chen2022maximum} we solve the problem in almost-linear time for the first time.
  
  \begin{restatable}{corollary}{almostLinear}
   Given a $k$-width DAG $G=(V,E)$ we can compute a minimum chain cover in time $O(|E|^{1+o(1)})$ w.h.p.
  \end{restatable}
  
  Moreover, our solution in \Cref{thm:mainTheorem} uses a MF solver as a black box and it is deterministic otherwise. Therefore, we provide a deterministic MCC solution in \emph{near-linear time} if one is found for MF. At the core of our solution we use mergeable dictionaries boosted with the \sizesplit{} operation to efficiently transform the flow outputted by the MF solver into an MCC. \emph{Mergeable dictionaries} is a data structure maintaining a dynamic partition $\partition$ of the natural numbers $\{1,\ldots,k\}$ (the reuse of $k$ is intentional as this is how our approach will use the data structure), starting from the partition $\partition = \{\{1,\ldots, k\}\}$ and supporting the following operations, for $K, K_1, K_2\in\partition$:
  \begin{itemize}
      \item[--] \search{}$(K, j)$: returns $\max_{i \in K, i \le j} i$ if any such element exists\footnote{This query is also known as \emph{predecessor} query in the literature.}.
      \item[--] \merge{}$(K_1,K_2)$: replaces $K_1$ and $K_2$ by $K_1\cup K_2$ in $\partition$.
      \item[--] \splitt{}$(K, j)$: replaces $K$ by $K'=\{i\in K\mid i\le j\}$ and $K\setminus K'$ in $\partition$.
  \end{itemize}
  
  Note that the \merge{} operation does not assume that $\max_{i\in K_1}i < \min_{i\in K_2}i$ (sets are non-overlapping) as generated by the \splitt{} operation. If the previous condition (or the analogous $\max_{i\in K_2}i < \min_{i\in K_1}i$) is assumed, the operation is known as \join{}. Mergeable dictionaries have applications in Lempel-Ziv decompression~\cite{farach1998string,bille2020decompressing,puglisi2019lempel}, mergeable trees~\cite{georgiadis2011data} and generalizations of union-find-split~\cite{lai2008complexity}. In this paper we show how to modify them to obtain a fast MCC algorithm. Next, we review the different approaches used to implement mergeable dictionaries in the literature.\\

\par{\noindent\textbf{Mergeable dictionaries.}} The first efficient implementation of mergeable dictionaries can be derived\footnote{The authors of~\cite{farach1998string} do not define mergeable dictionaries formally.} from the \emph{segment merge} strategy proposed by Farach and Thorup~\cite{farach1998string} to efficiently \merge{} two self-balanced binary search trees, assuming that operations \search{}, \splitt{} and \join{} are implemented in logarithmic time. The authors showed how to implement the \merge{} operation by minimally \splitt{}ing both sets into non-overlapping sets and pairwise \join{}ing the resulting parts. They proved that after $|V|+|E|$ operations (the abuse of notation is again intentional) their strategy works in $O(\log{k}\cdot\log{(|V|+|E|)})$ amortized time per operation. Later, Iacono and \"Ozkan~\cite{iacono2010mergeable} presented a mergeable dictionaries implementation running in $O(\log{k})$ amortized time per operation, based on biased skip lists~\cite{bagchi2005biased}. The same amortized running time was later achieved by Karczmarz~\cite{karczmarz2016simple} with a very simple approach using tries~\cite{de1959file} to represent the sets.\\

Our algorithm for MCC computes a minimum flow $f^*$ encoding an MPC and then extracts an MCC from $f^*$ by processing $G$ in topological order and querying  mergeable dictionaries boosted with the \sizesplit{} operation. Formally, we require the following operations, for $K,K_1,K_2\in\partition, s \in \{1,\ldots,|K|-1\}$:

\begin{itemize}
      \item[--] \some{}$(K)$: returns some element $i\in K$.
      \item[--] \merge{}$(K_1,K_2)$: replaces $K_1$ and $K_2$ by $K_1\cup K_2$ in $\partition$.
      \item[--] \sizesplit{}$(K, s)$: replaces $K$ by $K'\subseteq K, |K'| = s$ and $K\setminus K'$ in $\partition$.
  \end{itemize}

In \Cref{sec:mergeable-dict} we show how to implement these operations in $O(\log{k})$ amortized time each. Then, in \Cref{sec:almost-linear-mcc} we show our MCC algorithm using the previously described data structure. Besides the theoretical improvement already explained, we highlight the simplicity of our solutions, both in our proposal for boosted mergeable dictionaries as well as in our algorithm for MCC.

\section{Notation and preliminaries}\label{sec:preliminaries}

\par{\textbf{Graphs.}} For a vertex $v\in V$ we denote $N^{-}(v)$ ($N^{+}(v)$) to be the set of \emph{in(out)-neighbors} of $v$ that is $N^-(v) = \{u\in V\mid (u,v)\in E\}$ ($N^+(v) = \{w\in V\mid (v,w)\in E\}$). A $v_1v_\ell$-path is a sequence of vertices $P = v_1,\ldots,v_\ell$ such that $(v_i, v_{i+1}) \in E$ for $i
\in \{1,\ldots,\ell-1\}$, in this case we say that $v_1$ \emph{reaches} $v_\ell$. We say that $P$ is \emph{proper} if $\ell \ge 2$ and that $P$ is a cycle if $v_1 = v_\ell$. A \emph{directed acyclic graph} (DAG) is a graph without proper cycles. In a DAG we can compute, in linear time~\cite{kahn1962topological,tarjan1976edge}, a \emph{topological order} $v_1,\ldots,v_{|V|}$ of its vertices such that for every $i<j$, $(v_j, v_i) \not\in E$. In this paper we assume that $G$ is a DAG and since our algorithms run in time $\Omega(|V|+|E|)$, we assume that an input \emph{topological order} is given. A chain is a sequence of vertices $C = v_1,\ldots,v_{\ell'}$ such that for each $i
\in \{1,\ldots,\ell'-1\}$ $v_i$ reaches $v_{i+1}$. We denote $|C| = \ell'$ to the length of the chain. A \emph{chain cover} $\chaincover$ is a set of chains such that every vertex appears in some chain of $\chaincover$. We say that it is a \emph{chain decomposition} if every vertex appears in exactly one chain of $\chaincover$ and a \emph{path cover} if every chain of $\chaincover$ is a path, in this case we denote it $\pathcover$ instead. We denote $||\chaincover||$ to the total length of a chain cover that is $||\chaincover|| = \sum_{C\in\chaincover} |C|$.
An antichain is a subset of vertices $A\subseteq V$ such that for $u,v\in A, u\neq v$, $u$ does not reach $v$. The minimum size of a chain cover is known as the \emph{width} of $G$ and we denote it $k$. \\

\par{\noindent\textbf{Flows.}} Given a function of \emph{demands} $d:E\rightarrow \mathbb{N}_0$ and $s, t\in V$, an $st$-\emph{flow} is a function $f:E\rightarrow \mathbb{N}_0$ satisfying \emph{flow conservation} that is $inFlow_v := \sum_{u\in N^-(v)}f(u,v) = \sum_{w\in N^+(v)}f(v,w) := outFlow_v$ for each $v \in V\setminus \{s, t\}$, and \emph{satisfying the demands} that is $f(e) \ge d(e)$ for each $e\in E$. A \emph{flow decomposition} of $k$ (the reuse of notation is again intentional) paths of $f$ is a collection $\mathcal{D} = P_1,\ldots,P_k$ of $st$-paths such that for each edge $e\in E$, $f(e) = |\{P_i\in \mathcal{D}\mid e\in P_i\}|$.\footnote{This is a simplified definition of flow decomposition which suffices for our purposes.} The \emph{size} $|f|$ of $f$ is defined as the net flow entering $t$ (equivalently exiting $s$ by flow conservation) that is $|f| = inFlow_{t}-outFlow_t$. The problem of \emph{minimum flow} looks for a feasible $st$-flow of minimum size. Finding an MPC can be reduced to decompose a specific minimum flow~\cite{ntafos1979path}, we will revisit this reduction in \Cref{sec:almost-linear-mcc}. The same techniques applied for the problem of \emph{maximum flow} can be used in the context of the minimum flow problem~\cite{ciurea2004sequential,bang2008digraphs}. In fact, for MPC one can directly apply a maximum flow algorithm with \emph{capacities} at most $|V|$~\cite[Theorem 2.2 (full version)]{caceres2022sparsifying}.\\

\par{\noindent\textbf{Data structures.}} A \emph{self-balancing binary search tree} such as an AVL-tree or a red-black tree~\cite{guibas1978dichromatic} is a binary search tree supporting operations \search{}, \splitt{} and \join{} in logarithmic time each (in the worst case). A (binary) \emph{trie}~\cite{de1959file} is a binary tree representing a set of integers by storing their \emph{binary representation} as \emph{root-to-leaf} paths of the trie. In our tries all root-to-leaf paths have the same length $\lfloor \log{k} \rfloor + 1$. For a data structure supporting a set of operations, and a potential function $\phi$ capturing the state of the data structure, we say that the \emph{amortized} time of an operation equals to its (worst-case) running time plus the change $\Delta \phi$ in the potential triggered by the operation. If we apply a sequence of $O(|V|+|E|)$ operations whose amortized time is $O(\log{k})$ the total (worst-case) running time is $O((|V|+|E|)\log{k})$.

\section{Mergeable dictionaries with \sizesplit{}}\label{sec:mergeable-dict}

We show how to implement the \emph{boosted} mergeable dictionaries supporting operations \some{}, \merge{} and \sizesplit{} in $O(\log{k})$ amortized time each. To achieve this result we modify an existing solution of mergeable dictionaries implementing operations \search{}, \merge{} and \splitt{} by adding the \select{} operation. Formally for $K\in\partition, s \in {1,\ldots,|K|}$,

\begin{itemize}
    \item[--] \select{}$(K,s)$: returns the $s$-th smallest element in $K$.
\end{itemize}

With the \select{} operation we can use (normal) mergeable dictionaries to implement \some{} and \sizesplit{} as follows:

\begin{itemize}
    \item[--] \some{}$(K) \gets$ \search{}$(K, k)$.
    \item[--] \sizesplit{}$(K,s) \gets$ \splitt{}$(K, $\select{}$(K,s))$.
\end{itemize}

While for \some{} it suffices to do a \search{} with a known upper bound (recall that $k$ is the maximum element in the universe considered), in the case of \sizesplit{} we can first \select{} the corresponding pivot and use this pivot to \splitt{} the set by its value, obtaining the desired sizes for the split. 

This reduction allows us to obtain the boosted mergeable dictionaries by simply implementing the \select{} operation with logarithmic amortized cost (\sizesplit{} can be seen as one call to \select{} followed by a separate call to \splitt{}). Moreover, if the implementation of \select{} does not modify the data structure (and thus the potential $\phi$), its amortized time equals its (worst-case) running time (as $\phi$ does not change). We show that this is indeed the case in both the segment merge strategy of Farach and Thorup~\cite{farach1998string}, and the trie implementation of Karczmarz~\cite{karczmarz2016simple}, the latter achieving the desired running time.

The mergeable dictionaries based on segment merge represent each set as a self-balancing binary search tree. As such, the \select{} operation can be implemented in $O(\log{k})$ time by storing the sub-tree sizes at every node of the tree, which can be maintained (updated when the tree changes) in the same (worst-case) running time as the normal operations of the tree (see e.g.~\cite{knuth1998art}).

\begin{figure}
     \centering
     \begin{subfigure}[b]{0.48\textwidth}
         \centering
         \includegraphics[width=\textwidth]{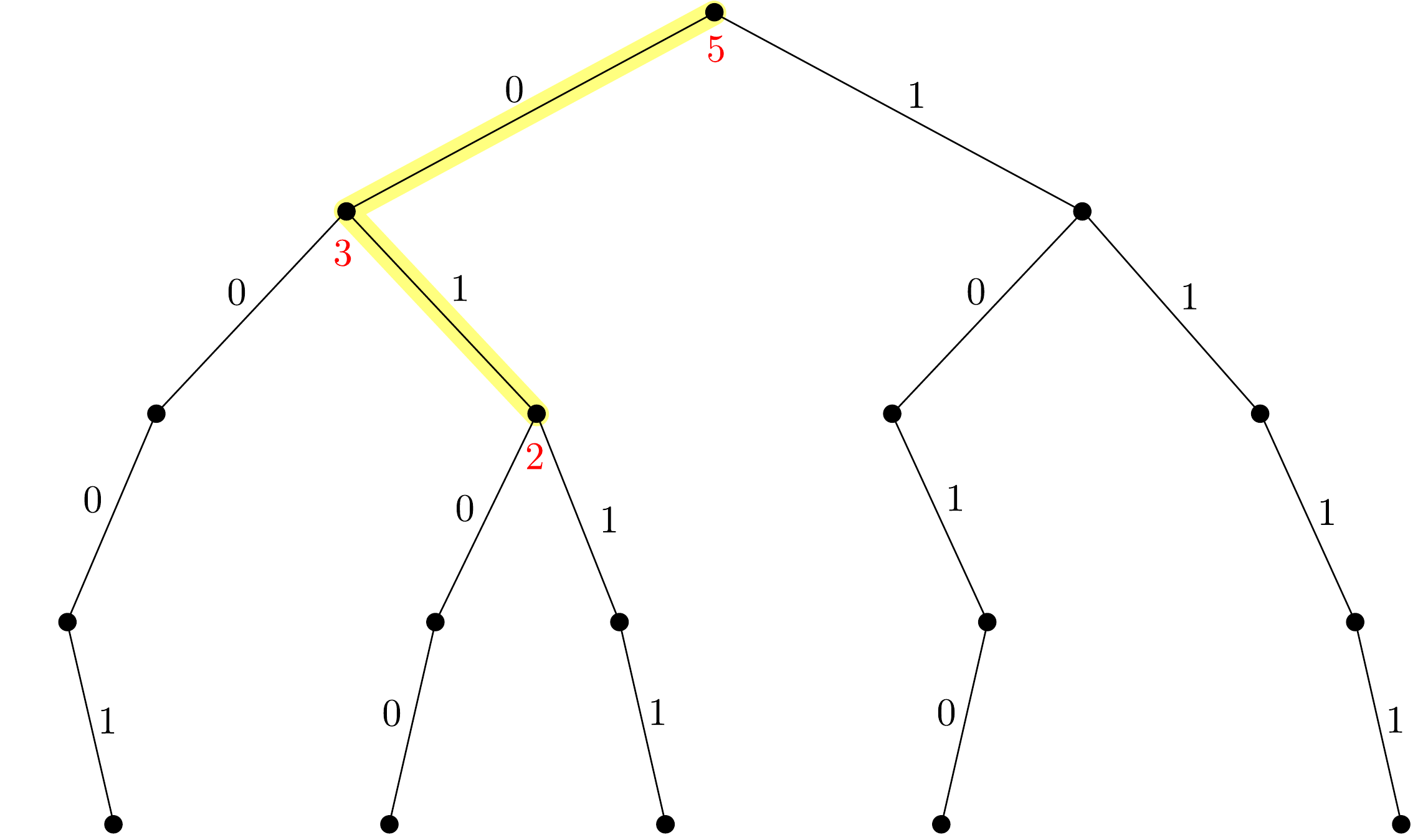}
     \end{subfigure}
     \hfill
     \begin{subfigure}[b]{0.48\textwidth}
         \centering
         \includegraphics[width=\textwidth]{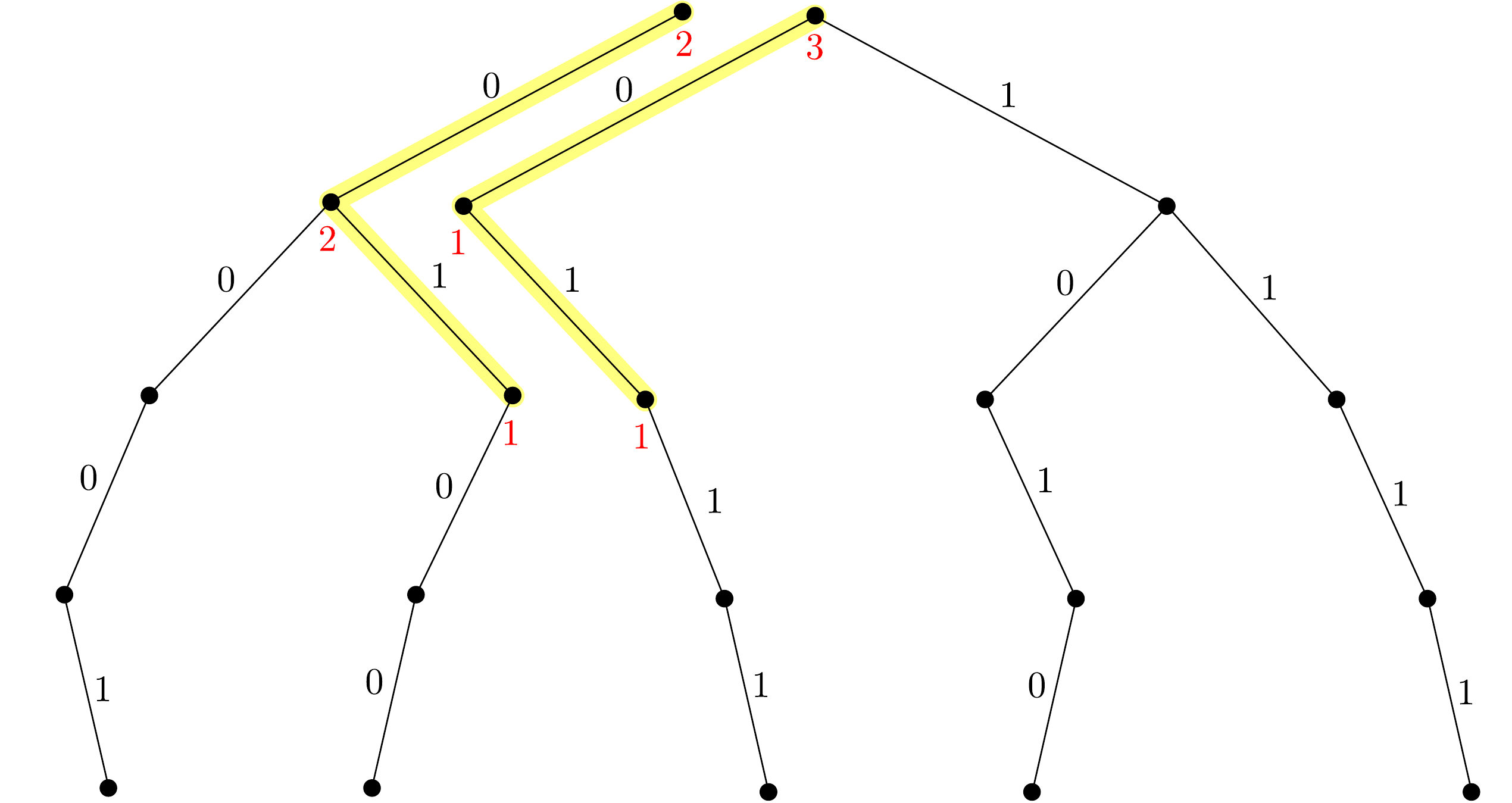}
     \end{subfigure}
        \caption{Result of calling \sizesplit{}$(K=\{1,4,7,10,15\},2)$ on a trie representation of boosted mergeable dictionaries. The binary representation of numbers in the set are spelled as root-to-leaf paths of the trie. Some number of leaves' counters are written in red under their respective nodes. Those counters are used to answer $4 \gets$ \select{}$(K, 2)$. After this, \splitt{}$(K,4)$ is performed. The path $P$ representing the longest prefix ($01$) between the binary representations of $4$ ($0100$) and $7$ ($0111$) is highlighted in yellow. The right figure shows the end result, note that only the counters in the nodes of $P$ change.}
        \label{fig:size-split-trie}
\end{figure}

Similarly, the implementation of Karczmarz~\cite{karczmarz2016simple} represents every set as a trie of its elements. As discussed in \Cref{sec:preliminaries}, a trie stores its elements as their binary representation encoded as (equal length) root-to-leaf paths of the trie. For example, if $k=6$ the corresponding binary representation of $3$ is $011$, which is represented as the root-to-leaf path following the left child, then its right child and then its right child.
Analogous to binary search trees, the nodes of a trie can be augmented to store the number of leaves in their respective sub-trees. If such augmentation of a trie is performed, then the operation \select{} can be implemented in $O(\log{k})$ (worst-case) time similar to the implementation on binary search trees, since the leaves in the trie follow the same order as the elements they represent: for a query \select{}$(s)$ on a trie node, it suffices to look at the number of leaves $l$ (elements) under the left child, if $l \ge s$ we continue to the left child answering \select{}$(s)$, otherwise we continue to the right child answering \select{}$(s-l)$.

Karczmarz~\cite{karczmarz2016simple} showed that \splitt{}$(K, j)$ can be performed by finding the root-to-node path $P$ corresponding to the longest common prefix between the binary representations of \search{}$(K,j)$ and of the smallest value greater than $j$ in $K$. It then splits the trie by removing the right children of the nodes of $P$ (to form $K'$) and then joining those nodes as right children of a copy of $P$ (to create $K\setminus K'$). Note that this procedure can be easily augmented to maintain the number of leaves under each node: only the nodes in $P$ (and in its copy) decrease their value by the number of leaves of their lost children. \Cref{fig:size-split-trie} shows an example of the \splitt{} operation. \merge{}$(K_1,K_2)$ is implemented as a simple recursive algorithm that at every step keeps one of the (common) nodes and then merges the corresponding left and right sub-trees (if one of those sub-trees is empty then it just keeps the other sub-tree, we refer to the original work \cite{karczmarz2016simple} for details). In this case the maintenance of the number of leaves can be performed when returning from the recursive calls: simply recompute the number of leaves as the sum of the number of leaves of their two children. Each of these computations is a constant time operation, and thus they do not change the asymptotic running time of \merge{}. We then obtain the following lemma.
 
 \begin{lemma}\label{lemma:mergeable-dict}
 There exists a data structure maintaining a dynamic partition $\partition$ of $\{1,\ldots,k\}$ starting from $\partition = \{\{1,\ldots, k\}\}$ and answering operations \some{}, \merge{} and \sizesplit{} such that for a sequence of $n=\Omega(k)$ operations\footnote{This requirement comes from the fact that mergeable dictionaries actually start from an empty collection of sets, however, one can create the singleton sets and merge then in total $O(k\log{k})$ time.} it answers in total $O(n\log{k})$ time.
 \end{lemma}
 \begin{proof}
     We first note that operations \some{}, \merge{} and \sizesplit{} can be implemented in the same asymptotic (worst-case) running time as operations \search{}, \merge{} and \splitt{} in the data structure of Karczmarz~\cite{karczmarz2016simple}, respectively. Indeed, as previously discussed, \some{} is implemented as one call to \search{}, \merge{} is implemented in the same way as in~\cite{karczmarz2016simple} but taking care of the number-of-leaves counters' updates which does not affect the asymptotic running time, and \sizesplit{} is implemented as one call to \select{} (in $O(\log{k})$ time) followed by one call to \splitt{} (also in $O(\log{k})$ time). For the amortized analysis we reuse the potential function used by Karczmarz~\cite{karczmarz2016simple}, namely, the number of nodes on all tries. Operations \some{} and \merge{} follow the same potential change as in~\cite{karczmarz2016simple} and thus have an $O(\log{k})$ amortized running time. Finally, since \select{} does not change the total number of nodes (nor any of the tries) the potential change of a \sizesplit{} is the same as the one of a \splitt{}, that is $O(\log{k})$ as in~\cite{karczmarz2016simple}. The lemma follows by using the amortized running times of the data structure's operations.
 \end{proof}

\section{An almost linear time algorithm for MCC}\label{sec:almost-linear-mcc}

We show how to use \Cref{lemma:mergeable-dict} to obtain a fast MCC algorithm. Our algorithm computes a minimum flow $f^*$ encoding an MPC of $G$ and then uses the data structure from \Cref{lemma:mergeable-dict} to efficiently extract an MCC from $f^*$. Next, we describe the well known~\cite{ntafos1979path} reduction from MPC to MF by following the notation of~\cite[Section 2.3 (full version)]{caceres2022sparsifying}.

\begin{figure}
      \centering
      \includegraphics[width=\textwidth]{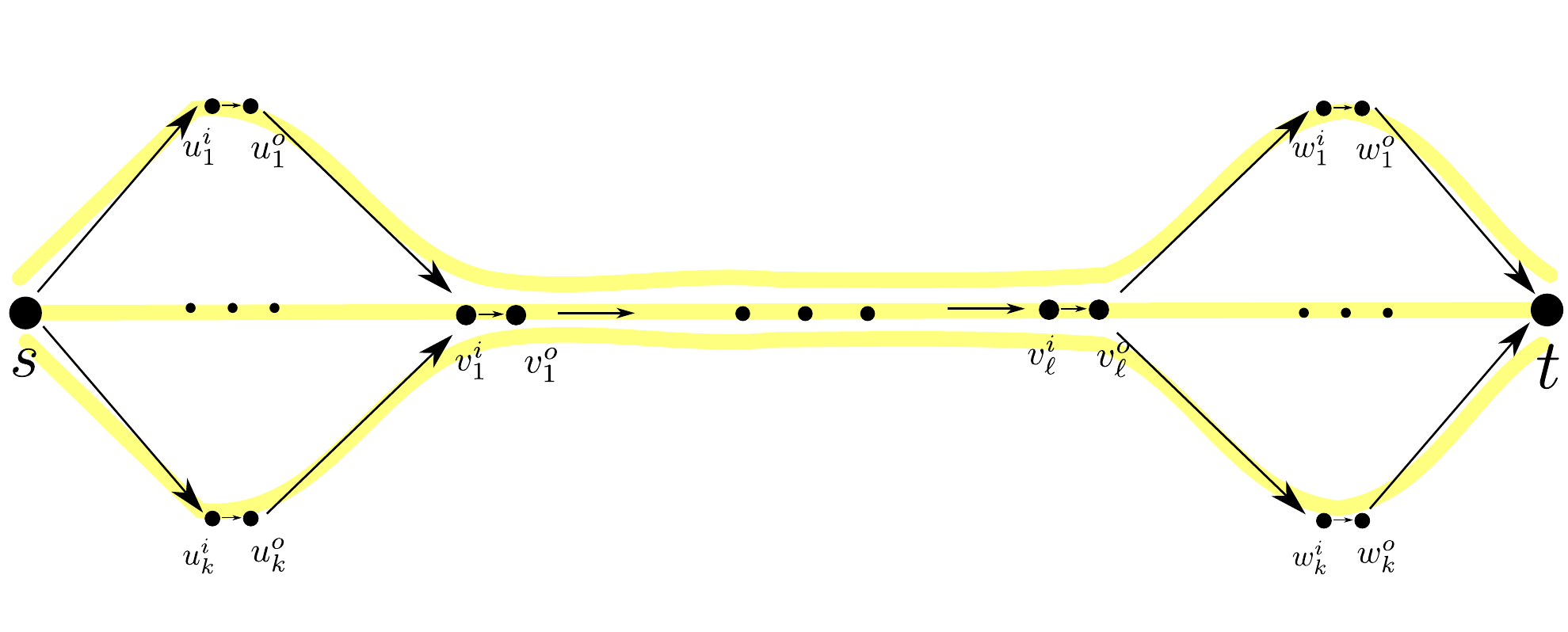}
      \caption{Flow reduction (demands are not shown) of the graph of \Cref{fig:worst-case-mpc} and a minimum flow on it. Only edges with positive flow are shown. The flow is implicitly presented as a flow decomposition showing every path highlighted in yellow, this flow decomposition corresponds to an MPC of the original DAG.}
      \label{fig:flow-reduction}
\end{figure}

Given the DAG $G$ we build its \emph{flow reduction} $\flowG = (\flowV, \flowE)$ as the graph obtained by adding a global source $s$, a global sink $t$ and splitting every vertex $v\in V$ into two copies connected by an edge. Additionally, the first copy, $v^i$, is connected from the in-neighbors of $v$ and the second copy, $v^o$, is connected to the out-neighbors of $v$. Formally, $\flowV = \{s,t\} \cup \{v^i \mid v\in V\} \cup \{v^o \mid v\in V\}$, and $\flowE =  \{(s, v^{i}) \mid v \in V\} \cup \{(v^{o}, t) \mid v \in V\} \cup \{(v^{i}, v^{o}) \mid v \in V\} \cup \{(u^{o}, v^{i}) \mid (u,v) \in E\}$. Note that $|\flowE| = O(|V|+|E|)$, and that $\flowG$ is also a DAG. We also define demands on the edges, $d: \flowE: \rightarrow \mathbb{N}_0$, as $1$ if the edge is of the form $(v^{i}, v^{o})$, and $0$ otherwise. Intuitively, the demands \emph{require} that at least one unit of flow goes through every vertex (of $G$), which directly translates into to the path cover condition of covering each vertex with at least one path. In fact, every flow decomposition (recall \Cref{sec:preliminaries}) of a feasible $st$-flow $f$ of $\flowG, d$ corresponds to a path cover of $G$ of size $|f|$. Moreover, every decomposition of a minimum flow $f^*$ of $\flowG, d$ corresponds to an MPC of $G$, and thus $k = |f^*|$~\cite[Section 2.3 (full version)]{caceres2022sparsifying}. \Cref{fig:flow-reduction} illustrates these ideas with an example. 

Since every vertex cannot belong to more than $|V|$ paths in an MPC, the problem can be reduced to maximum flow~\cite[Theorem 3.9.1]{bang2008digraphs}. We summarize these results in the following lemma.

\begin{lemma}[Adaptation of {\cite[Theorem 2.2 (full version)]{caceres2022sparsifying}}]\label{lemma:minFlow}
We can compute a flow $f^*$ of $\flowG, d$ such that every flow decomposition of $f^*$ corresponds to an MPC of $G$, in time $O(T_{MF}(|E|))$, where $T_{MF}(|E|)$ is the time for solving maximum flow.
\end{lemma}

\begin{algorithm}[t]
\DontPrintSemicolon
\KwIn{A directed acyclic graph $G=(V,E)$.}
\KwOut{A minimum chain decomposition $\chaincover = C_1, \ldots, C_k$ of $G$.}\vspace{0.5em}

    $(\flowG , d) \gets$ Build the \emph{flow reduction} of $G$ detailed in \Cref{sec:almost-linear-mcc}\;
    $f^*, k = |f^*| \gets$ Use \Cref{lemma:minFlow} to obtain a minimum flow of $(\flowG , d)$\;
    
    Initialize $C_i$ as an empty list for $i \in \{1,\ldots,k\}$\;
    
    $I_s \gets \{1,\ldots,k\}$\;

    \For{$v \in V$ in topological order}{
        $I_v \gets$ Take $f^*(s, v^i)$ elements from $I_s$\;
        \For{$u^o \in N^{-}(v^i)$}{
            $I_{uv} \gets$ Take $f^*(u^o, v^i)$ elements from $I_u$\;
            $I_v \gets I_v \cup I_{uv}$
        }
        $i \gets$ Choose an element from $I_v$\;
        $C_i$.append($v$)\;
    }
    \Return $C_1,\ldots,C_k$\;
 \caption{\label{alg:simpleVersion}Non-optimized pseudocode for our MCC algorithm. A naive implementation of this algorithm obtains an $O(||\pathcover||)$ running time as explained in this manuscript.}
\end{algorithm}

A decomposition algorithm is simple in this case: start from $s$ and follow a path $P$ of positive flow edges until arriving at $t$, and then update $f^*$ decreasing the flow on the edges of $P$ by one. Repeat this process until no flow remains. The $st$-paths obtained during the decomposition can be easily transformed into an MPC $\pathcover$ of $G$ (trim $s$ and $t$ and replace $v^i,v^o$ by $v$ on each path, see e.g.~\cite{kogan2023faster}).

If implemented carefully (see e.g.~\cite[Lemma 1.11 (full version)]{kogan2022beating}), the previous algorithm runs in time $O(||\pathcover||)$ and it outputs a valid MCC $\pathcover$ (recall that every MPC is an MCC). However, as shown in \Cref{fig:worst-case-mpc}, $||\pathcover||$ can be $\Omega(k|V|)$ in the worst case. Our algorithm overcomes this barrier by instead directly extracting (from $f^*$) a minimum chain decomposition (MCD, recall \Cref{sec:preliminaries}) and thus its total length is exactly $|V|$.

The main idea to extract an MCD $\chaincover = C_1,\ldots,C_k$ from $f^*$ is to compute, for each vertex $v\in\flowV$, the set $I_v \subseteq \{1,\ldots,k\}$ of indices such that $v$ would belong to paths $\pathcover_v = \{P_i \mid i \in I_v\}$ in a flow decomposition $\mathcal{D} = P_1, \ldots, P_k$ of $f^*$. Note that $I_s = I_t = \{1,\ldots,k\}$ by construction of $\flowG, d$. To efficiently compute these sets, we process the vertices in a topological order of $\flowG$, for example $s, v_1^i, v_1^o,\ldots,v_{|V|}^i,v_{|V|}^o,t$ (recall that a topological order $v_1,\ldots,v_{|V|}$ of $G$ is assumed as input).  When processing vertex $v$ we compute $I_v$ as follows: for every $u\in N^-(v)$ we \emph{take} (exactly) $f^*(u, v)$ elements from $I_{u}$, let us denote $I_{uv}$ to these elements. Then, we compute $I_{v}$ as the \emph{union} $\bigcup_{u \in N^-(v)} I_{uv}$. And finally, we \emph{take an arbitrary element} $i \in I_v$ and append $v$ to $C_i$. \Cref{alg:simpleVersion} shows a pseudocode for this algorithm.

\begin{algorithm}[t]
\DontPrintSemicolon
\KwIn{A directed acyclic graph $G=(V,E)$.}
\KwOut{A minimum chain decomposition $\chaincover = C_1, \ldots, C_k$ of $G$.}\vspace{0.5em}

    $(\flowG , d) \gets$ Build the \emph{flow reduction} of $G$ detailed in \Cref{sec:almost-linear-mcc}\;
    $f^*, k = |f^*| \gets$ Use \Cref{lemma:minFlow} to obtain a minimum flow of $(\flowG , d)$\;
    
    Initialize $C_i$ as an empty list for $i \in \{1,\ldots,k\}$\;
    
    Initialize mergeable dictionaries maintaining a partition of $\{1,\ldots,k\}$\;
    
    $I_s \gets \{1,\ldots,k\}$\;

    \For{$v \in V$ in topological order}{
        $I_v, I_s \gets$ \sizesplit{}$(I_s, f^*(s, v^i))$\;
        \For{$u^o \in N^{-}(v^i)$}{
            $I_{uv}, I_{u} \gets$ \sizesplit{}$(I_u, f^*(u^o, v^i))$\;
            $I_v \gets$ \merge{}$(I_v, I_{uv})$\;
        }
        $i \gets$ \some{}$(I_v)$\;
        $C_i$.append($v$)\;
    }
    \Return $C_1,\ldots,C_k$\;
 \caption{\label{alg:pseudocode} Our MCC algorithm running in time $O(T_{MF}(|E|)+ (|V|+|E|)\log{k})$, where $T_{MF}(|E|)$ is the time for solving maximum flow. The algorithm uses the boosted mergeable dictionaries from \Cref{sec:mergeable-dict}.}
\end{algorithm}

Note that vertices are added to their respective chains in the correct order since they are processed in topological order. 

Moreover, since indices are moved from one vertex to the other only if there is an edge between them, consecutive vertices are always connected by a path in $G$, and thus the lists $C_i$ correspond to proper chains. Finally, adding each vertex to only one such chain ensures that the end result is indeed an MCD (every vertex in exactly one chain).

An important aspect of the algorithm is that exactly $f^*(u^i,v^i)$ ($f^*(s,v^i)$) elements are \emph{taken out} of $I_u$ ($I_s$). This step is always possible thanks to flow conservation of $f^*$.

\begin{lemma}\label{lemma:correctness}
    Given a $k$-width DAG $G = (V, E)$ as input, \Cref{alg:simpleVersion} computes a minimum chain decomposition $\chaincover = C_1, \ldots, C_k$ of $G$.
\end{lemma}
\begin{proof}
    By \Cref{lemma:minFlow}, every flow decomposition of $f^*$ corresponds to an MPC of $G$. We prove that each $C_i$ is a chain of $V$, since every vertex is added to exactly one $C_i$ we conclude that $\chaincover = C_1, \ldots, C_k$ is a minimum chain decomposition. Inductively, if $v$ is added after $v'$ in chain $C_i$, then $v'$ reaches $v$ in $G$. Indeed, $i \in I_v$ and in particular $i \in I_u$ for some $u \in N^-(v)$, and inductively $v'$ reaches $u$ in $G$ and thus also reaches $v$.
\end{proof}

Moreover, since only splits and unions are performed, the sets $I_v$'s and $I_{uv}$'s form a partition of $\{1,\ldots,k\}$ at any point during the algorithm's execution.

A simple implementation of sets $I_v$'s and $I_{uv}$'s as linked lists, allows us to perform the unions and element picks in constant time, but the splits in $O(f^*(u^o,v^i))$ (and $O(f^*(s,v^i))$) time each, and thus in $O(||\pathcover||)$ time in total. However, we can implement the sets $I_v$'s and $I_{uv}$'s as boosted mergeable dictionaries from \Cref{sec:mergeable-dict} to speed up the total running time. \Cref{alg:pseudocode} shows the final result.

\mainTheorem*
\begin{proof}
The correctness of the algorithm follows from the previous discussion and \Cref{lemma:correctness} since \Cref{alg:pseudocode} is an implementation of \Cref{alg:simpleVersion}. Building the flow reduction takes $O(|V|+|E|)$ time and obtaining the minimum flow $f^*$ takes $O(T_{MF}(|E|))$ time by \Cref{lemma:minFlow}. The rest of the running time is derived from the calls to the mergeable dictionaries' operations. \some{} is called $O(|V|)$ times while \sizesplit{} and \merge{} are called once per edge in the flow reduction that is $O(|\flowE|) = O(|V|+|E|)$ times. By applying \Cref{lemma:mergeable-dict} the total time of the $O(|V|+|E|)$ operation calls is $O((|V|+|E|)\log{k})$.
\end{proof}

We finish our paper by encapsulating our result into a tool that can be used to efficiently extract a set of vertex-disjoint chains $\chaincover$, encoding a set of paths $\pathcover$, from a flow $f$ that encodes $\pathcover$ as a flow decomposition. If the problem can be modeled as a maximum flow/minimum cost flow problem, the result of Chen et~al.~\cite{chen2022maximum} allows us to solve such problems in almost linear time. As a simple example, we could solve the $\ell$-cover problem (find a set of $\ell$ vertex-disjoint chains covering the most vertices) in almost linear time.

\begin{corollary}
Let $G=(V,E)$ be a DAG, and $f:E\rightarrow \mathbb{N}_0$ a flow encoding a set of $|f|$ paths $\pathcover$ of $G$ as a flow decomposition into weight-$1$ paths. In $O((|V|+|E|)\log{|f|})$ time, we can compute a set of $|f|$ vertex-disjoint chains $\chaincover$ of $G$, which can (alternatively) be obtained by removing repeated vertices from $\pathcover$.
\end{corollary}

\bibliographystyle{plain}
\bibliography{references}
\end{document}